\documentclass[10pt,conference]{IEEEtran}

\usepackage[left=0.625in,right=0.625in,top=0.75in,bottom=1in]{geometry}
\usepackage{booktabs}
\usepackage{graphics,colortbl,multicol,lipsum,xcolor,graphicx,color}
\usepackage[cmex10]{amsmath}
\usepackage{amsthm,mathrsfs,mathtools,amsbsy,amsfonts,amssymb,}
\usepackage{bbm}
\usepackage{setspace,float,pstool,lettrine,newclude}
\usepackage[normalem]{ulem}
\usepackage{enumitem}
\usepackage{latexsym,algpseudocode,gensymb,balance,multirow,bm,wasysym}
\usepackage{hhline}
\usepackage{cite}
\usepackage[linesnumbered,ruled,vlined]{algorithm2e}
\SetKwInput{KwInput}{Input}  
\SetKwInput{KwOutput}{Output}
\SetKw{KwBy}{by}
\makeatletter
\newcommand{\nosemic}{\renewcommand{\@endalgocfline}{\relax}}
\newcommand{\dosemic}{\renewcommand{\@endalgocfline}{\algocf@endline}}
\let\oldnl\nl
\newcommand{\nonl}{\renewcommand{\nl}{\let\nl\oldnl}}
\makeatother

\definecolor{pull}{RGB}{118, 223, 191}
\definecolor{silent}{RGB}{248,96,113}
\ifCLASSOPTIONcompsoc
\usepackage[caption=false,font=normalsize,labelfon
t=sf,textfont=sf]{subfig}
\else
\usepackage[caption=false,font=footnotesize]{subfig}
\fi

\DeclarePairedDelimiterX\MeijerM[3]{\lparen\!}{\rparen}%
{\,#3\delimsize\vert\begin{smallmatrix}#1 \\ #2\end{smallmatrix}}

\newcommand\MeijerG[8][]{%
  G^{\,#2,#3}_{#4,#5}\MeijerM[#1]{#6}{#7}{#8}}

\WithSuffix\newcommand\MeijerG*[7]{%
  G^{\,#1,#2}_{#3,#4}\MeijerM*{#5}{#6}{#7}}

\allowdisplaybreaks

\newtheorem{proposition}{Proposition}

\graphicspath{{figures/}}
\allowdisplaybreaks

\newcommand{\RNum}[1]{\uppercase\expandafter{\romannumeral #1\relax}}

\usepackage{mathtools}

\usepackage{accents}

\IEEEoverridecommandlockouts

\begin{document}

\title{Effective Communication: When to Pull Updates?}
\author{\IEEEauthorblockN{Pouya Agheli\IEEEauthorrefmark{1}, Nikolaos Pappas\IEEEauthorrefmark{2}, Petar Popovski\IEEEauthorrefmark{3}, and Marios Kountouris\IEEEauthorrefmark{1}}
    \IEEEauthorblockA{\IEEEauthorrefmark{1}Communication Systems Department, EURECOM, Sophia Antipolis, France}
    \IEEEauthorblockA{\IEEEauthorrefmark{2}Department of Computer and Information Science, Link\"{o}ping University, Link\"{o}ping, Sweden}
    \IEEEauthorblockA{\IEEEauthorrefmark{3}Department of Electronic Systems, Aalborg University, Aalborg, Denmark\\ Email: pouya.agheli@eurecom.fr, nikolaos.pappas@liu.se, 
    petarp@es.aau.dk,
    marios.kountouris@eurecom.fr}
    \thanks{The work of P. Agheli and M. Kountouris has received funding from the European Research Council (ERC) under the European Union’s Horizon 2020 research and innovation programme (Grant agreement No. 101003431). The work of N. Pappas is supported by the VR, ELLIIT, and the European Union (ETHER, 101096526), and the work of P. Popovski was supported by the Villum Investigator Grant ``WATER" from the Velux Foundation, Denmark.}
}
\maketitle

\begin{abstract}
We study a pull-based communication system where a sensing agent updates an actuation agent using a query control policy, which is adjusted in the evolution of an observed information source and the usefulness of each update for achieving a specific goal. For that, a controller decides whether to pull an update at each slot, predicting what is probably occurring at the source and how much effective impact that update could have at the endpoint. Thus, temporal changes in the source evolution could modify the query arrivals to capture important updates. The amount of impact is determined by a \emph{grade of effectiveness} (GoE) metric, which incorporates both freshness and usefulness attributes of the communicated updates. Applying an iterative algorithm, we derive query decisions that maximize the long-term average GoE for the communicated packets, subject to cost constraints. Our analytical and numerical results show that the proposed query policy exhibits higher effectiveness than existing periodic and probabilistic query policies for a wide range of query arrival rates.
\end{abstract}

\IEEEpeerreviewmaketitle

\section{Introduction}\label{Sec:Sec1}
The effectiveness problem in communication systems deals with whether or not a message conveyed by a sender, leads to a desirable impact at the receiver for achieving a specific goal. This has been first articulated in \cite{ShannonWeaver49}, and has recently been revived under the prism of goal-oriented semantic communication \cite{kountouris2021semantics,popovski2020semantic}.
In cyber-physical systems with interacting sensing entities and actuation/monitoring agents, a message ought to be generated and communicated \emph{if} it can potentially have the desired impact in the system. In this setting, the communication goal determines the \emph{grade of effectiveness} a message has according to its usefulness or importance in fulfilling a set of attributes required for achieving that goal. This approach has the potential to enable system scalability and judicious use of resources by avoiding the acquisition, processing, and transport of information that turns out to be ineffective, irrelevant, or useless.
 
In current networked intelligent systems, information transfer, e.g., in the form of status update packets, over the network is mainly done using a \emph{push-based} communication model. Packets are pushed toward the endpoint according to the source’s decision, regardless of what and when the endpoint actually needs. On the other hand, we have the \emph{pull-based} communication framework, where the endpoint requests and controls the type of the generated information and its arrival time \cite{pullB,pullC, pullD, pullE, pullF, yavascan2021pull, ildiz2021pull}. To this end, existing query policies overlook the evolution of the source (what is probably happening on the sensing side) and the expected importance or usefulness of the updates at the time of pulling them. As a result, query-based policies, even if they are aligned with the receiver's objectives, do not necessarily result in high system effectiveness. To address this challenge and to bridge the pull-based with the push-based world, we introduce a novel policy where the endpoint pulls updates based on the statistics of the source and the previously received updates. Thanks to this policy, the endpoint timely adapts its query instances according to the received updates, hence being able to capture updates (e.g., critical source events, alarms, novelty, etc.) that have high impacts on achieving the goal. In a way, our policy can be seen as an adaptive query policy where requests tend to align to the most useful or important system-wise source realizations.   


This paper falls within the realm of pull-based communication methods, but with a query control policy, owing to which the query controller timely decides to pull updates following the source changes and the updates' expected usefulness. Our work leverages the concept of Just-in-Time systems \cite{leiserson1979systolic} and extends prior work on query age of information (QAoI) \cite{pullB,pullC, pullD, pullE, pullF, yavascan2021pull, ildiz2021pull} into the systems where queries arrive at the right times to request the source realizations (updates) with high usefulness and significant effect at the endpoint. In this regard, a class of optimal policies is derived to maximize the long-term expected effectiveness of the update packets a sensing agent sends to an acting agent, subject to a communication cost. Our analytical and simulation results show that the proposed query policy outperforms existing baseline approaches in terms of effectiveness grade. Furthermore, we demonstrate that the solution converges to a threshold-based query control scheme, where the query controller can timely decide whether to pull an update or not by merely looking up a table with the obtained threshold boundaries for a given goal.

\section{System Model}\label{Sec:Sec2}
We consider an end-to-end pull-based communication system in which update packets are generated via a \emph{sensing agent} (SA) and transmitted to an \emph{actuation agent} (AA) for taking appropriate actions at the endpoint and fulfilling the subscribed goal. The packet generation is done in the event of receiving a \emph{query} from a query controller at the AA. Here, the packets contain status updates of a sensor observing an information source/event or payload data from the application layer towards the network. We assume the system operates in a slotted manner with time slots $n$, for $n \in \mathbb{N}$. At the $n$-th time slot, the query arrival indicator is denoted by $\alpha_n = \{0, 1\}$, with $\alpha_n = 1$ indicating the arrival of a query, and $\alpha_n = 0$, otherwise. Thus, the set of slots at which queries arrive is defined as $\mathcal{N} = \{n \,|\, \forall n : \alpha_n = 1\}$. The update and query channels are assumed to be \emph{error-prone} and \emph{error-free}, respectively, and the whole duration of receiving a query and the consequent update acquisition and its communication is normalized into one time slot. Also, we consider $x_n$ the update communicated at slot $n$, which can be seen as some form of semantic representation.

In parallel with its acquisition, the $n$-th update is evaluated based on its usefulness for satisfying the goal and is attributed a rank of importance (value) at the source level in the form of a meta-value $v_n$, $\forall n \in \mathcal{N}$ \cite[Section~III-A]{agheli2023semantic}. For the sake of generality, we assume that $v_n$ is a random variable (r.v.) that follows a discrete-time Markov process with finite state space $\mathcal{V} = \{\nu_i \,|\, i = 1, 2, ..., |\mathcal{V}|\}$ and transition matrix $\mathbf{P}_v = [p_{ij}]_{|\mathcal{V}|\times|\mathcal{V}|}$ with $p_{ij} = \mathrm{Pr}(v_{n+1}=\nu_j\,|\,v_n = \nu_i)$. From an effectiveness viewpoint, a successfully received update has a specific impact at the receiver or decision-making side. In the proposed model, we consider that the goal is perfectly known at both SA and AA, thus the same framework is employed to measure the usefulness of updates at both ends.\footnote{The analysis can be easily extended to the case where the goal is not shared, but the AA, based on the received updates, may learn or estimate the probability distribution of the updates’ usefulness at the endpoint.}

\subsection{Query Control Policies}
We consider that the query controller can apply two different types of policies to pull updates from the SA in a pull-based system, as follows.

\subsubsection{Effect-agnostic policy}
Under this policy, the queries arrive with a \emph{controlled query rate}, pursuing a specific schedule or a stochastic process, e.g., Poisson, binomial, and Markovian \cite{pullC, pullD, pullE, pullF}. Therefore, the existing effect-agnostic policies pose \emph{aleatoric} uncertainty associated with the nature of random updates by ignoring what is probably occurring at the source during the time of pulling those updates.

\subsubsection{Effect-aware policy}
With this policy, the AA tries to infer or predict the probable state of the source and the expected usefulness of an update at the time of the decision, hence adapting its query instances and pulling updates to the right time slots. This can be done via building or updating a model of the source’s process at the AA. Owing to this policy, albeit packet generation and communication are requested by queries, what happens at the source is also considered. Therefore, \emph{epistemic} uncertainty arises with the effect-aware policy, as decisions are made based on probabilistic predictions, instead of accurate knowledge. This uncertainty could be reduced or harnessed using for instance machine learning or prediction methods. We delve into this query control problem in Section~\ref{Sec:Sec4}.



\section{Grade of Effectiveness Metric}\label{Sec:Sec3}
To model effectiveness, we advocate for a metric that involves two system-level attributes: \emph{freshness} and \emph{usefulness} of the successfully received updates. Freshness describes how obsolete an update gets as time passes from the instant it is correctly received at the AA. This attribute is commonly quantified via the age of information (AoI) metric or its variants. Besides, the usefulness signifies the rank of importance, as defined in Section~\ref{Sec:Sec2}, each correctly received packet offers at the endpoint.

With the above explanation in mind, we propose the grade of effectiveness (GoE) metric for measuring the amount of impact an update has at the endpoint.

\vspace{-0.05cm}
\subsection{GoE Formulation}
The GoE for the update packet communicated at the $n$-th time slot is denoted by $\operatorname{GoE}_n \in \mathbb{R}_0^+$ and modeled in the form of a composite function $f: \mathbb{R}_0^{+}\times\mathbb{R}_0^{+} \rightarrow \mathbb{R}_0^+$ of the AoI, called $\Delta_n$ for the $n$-th time slot, and the usefulness of that update, i.e., $v_n$, to satisfy the subscribed goal. Thus, we can write\footnote{The GoE metric can be seen as a particular case of the Semantics of Information (SoI) metric \cite{kountouris2021semantics,pappas2021goal} and can also be defined in various forms based on the scenario, including that utilized for active fault detection in \cite{stamatakis2022semantics}.}
\begin{align}\label{eq:eq1}
\operatorname{GoE}_n &= f\big(g_\Delta(\Delta_n), g_v(v_n)\big)
\end{align} 
where $g_\Delta : \mathbb{R}_0^+ \rightarrow \mathbb{R}_0^+$ is a non-increasing \emph{penalty} function, and $g_v : \mathbb{R}_0^+ \rightarrow \mathbb{R}_0^+$ shows a non-decreasing \emph{utility} function. In \eqref{eq:eq1}, $\Delta_n = n-m$ with $\Delta_0 = 1$, and $m = \underset{i :\,\alpha_i(1-\epsilon_i)=1,\, i \leq n}{\rm {max}~} i$. Also, $\epsilon_n = g_\epsilon(d(x_n, \hat{x}_n))$ depicts the update discrepancy, where $\hat{x}_n$ indicates the received update at the AA, $g_\epsilon : \mathbb{R}_0^+ \rightarrow \{0, 1\}$ is the mapping function to the Boolean space, and $d : \mathbb{R} \rightarrow \mathbb{R}_0^+$ is a distance function.
In \eqref{eq:eq1}, if we overlook the usefulness of updates or their freshness, the GoE metric turns into the penalty definition for the QAoI, i.e., $\operatorname{GoE}_n = f\big(g_\Delta(\Delta_n)\big)$, or the utility formulation for the \emph{value of information} (VoI), i.e., $\operatorname{GoE}_n = f\big(g_v(v_n)\big)$, respectively.

\vspace{-0.05cm}
\subsection{The Pull-Based System's GoE}
Since the update acquisition and communication is done on the condition of receiving queries in the pull-based system, $\operatorname{GoE}_n^{\rm (pull)}$ becomes limited to the slots at which the queries arrive, i.e., $\forall n \in \mathcal{N}$. From \eqref{eq:eq1}, we can write
\begin{align}\label{eq:eq2}
\operatorname{GoE}_n^{\rm (pull)} = \operatorname{GoE}_{n \,|\, \alpha_n=1} \times \mathbbm{1}(\alpha_n = 1)
\end{align}
where $\mathbbm{1}(\cdot)$ denotes the indicator function.

\vspace{-0.05cm}
\section{Query Control}\label{Sec:Sec4}
In this section, we first formulate the query control problem for the effect-aware policy, and we then propose a solution for the defined decision/control problem.

\vspace{-0.05cm}
\subsection{Problem Formulation}\label{Sec:Sec4a}
The objective is to maximize the long-term expected GoE via controlling the query arrivals, subject to a constraint on the average communication cost $C_{\rm max}$, which cannot be surpassed. In this sense, a class of optimal policies, named $\pi^*$, is derived by solving an optimization problem as follows
\begin{align}\label{opt:opt1}
    \mathcal{P}_1 : ~ &\underset{\pi}{\operatorname{max}} ~\underset{N \rightarrow \infty}{\lim \sup} ~ \frac{1}{N}\, \mathbb{E}\bigg[ \sum_{n = 1}^{N} \operatorname{GoE}_n^{\rm (pull)} \big| \operatorname{GoE}_0^{\rm (pull)}\bigg] \nonumber \\
    & {\rm s.t.}~ \underset{N \rightarrow \infty}{\lim \sup} ~ \frac{1}{N}\, \mathbb{E}\bigg[ \sum_{n = 1}^{N} g_c(\alpha_n c_n)\bigg] \leq C_{\rm max}
\end{align}
where $g_c : \mathbb{R}_0^+ \rightarrow \mathbb{R}_0^+$ is a non-decreasing function, and $c_n$ indicates the communication cost at the $n$-th time slot.

We cast $\mathcal{P}_1$ into an \emph{infinite-horizon} constrained Markov decision process (CMDP) based on the following definitions:
\begin{enumerate}
    \item[I)] \textit{States:}
    The state at the $n$-th slot depicts the GoE and is denoted with a tuple $S_n = (\Delta_n, v_n)$. 
    Without loss of generality, we consider the AoI to be truncated by the maximum value of $\Delta_{\rm max}$, which is large enough to represent excessive staleness and meet $g_\Delta(\Delta_{\rm max}\!-\!1)\leq (1 +\varepsilon) g_\Delta(\Delta_{\rm max})$ with accuracy $\varepsilon$. 
    Given this, $S_n$ is a member of the state space $\mathcal{S} = \{\sigma_i\,|\, i = 1, 2, ..., |\mathcal{S}|\}$, which is countable and finite with $|\mathcal{S}| = \Delta_{\rm max}\cdot|\mathcal{V}|$.
    \item[II)] \textit{Actions:}
    Regarding Section~\ref{Sec:Sec2}, we define the action space $\mathcal{A}$ with two possible outcomes: $\alpha_n = 1$ for pulling an update, and $\alpha_n = 0$ for keeping silent. 
    \item[III)] \textit{Transition probabilities:} The transition probability from $S_n$ to $S_{n+1}$ under the action $\alpha_n$ is $P(S_n, \alpha_n, S_{n+1}) = \mathrm{Pr}((\Delta_{n+1}, v_{n+1})\,|\,(\Delta_{n}, v_n), \alpha_n)$, $\forall n$, where we define
    \begin{itemize}
        \item $\mathrm{Pr}((1, \nu_j)\,|\,(\Delta_{n}, \nu_{i}), \alpha_n) = p_{ij}\alpha_n(1\!-\!p_\epsilon^{(n)})$, $\forall i,j$,
        \item $\mathrm{Pr}((\operatorname{min}\{\Delta_{n}\!+\!1, \Delta_{\rm max}\}, \nu_{i})\,|\,(\Delta_{n}, \nu_{i}), \alpha_n) = \alpha_n p_\epsilon^{(n)} + (1\!-\!\alpha_n)$, $\forall i$.
    \end{itemize}
    Otherwise, we have $P(S_n, \alpha_n, S_{n+1})=0$.
    Here, $\nu_{i}, \nu_{j} \in \mathcal{V}$, and $p_\epsilon^{(n)} = \mathrm{Pr}(\epsilon_n = 1)$. 
    \item[IV)] \textit{Rewards:} 
    The reward of going from $S_n$ to $S_{n+1}$ in the course of the action $\alpha_n$ is $R(S_n, \alpha_n, S_{n+1}) = \operatorname{GoE}_{n+1}^{\rm (pull)}$.
\end{enumerate}

With the above definitions in mind, we give the following proposition that states that the expected sum of GoE in \eqref{opt:opt1} is the same for all initial states, hence there exists an optimal stationary policy for the defined problem. 

\begin{proposition}\label{propose1} The modeled CMDP pursues the weak accessibility (WA) condition.
\end{proposition}
\begin{proof}
We can partition the space set $\mathcal{S}$ into two subsets $\mathcal{S}_a = \{S_n\in\mathcal{S}\,|\,S_n=(1, \nu_n), \forall \nu_n\in\mathcal{V}\}$ and $\mathcal{S}_b = \mathcal{S} - \mathcal{S}_a$. The primary subset contains all states whose $\Delta_n=1$, hence $|\mathcal{S}_a|=|\mathcal{V}|$. The latter subset includes the rest of the states with $\Delta_n\geq 2$, and $|\mathcal{S}_b| = (\Delta_{\rm max} \!-\! 1)\cdot|\mathcal{V}|$. In this regard, all states of $\mathcal{S}_b$ are transient under any stationary policy, while every state of an arbitrary pair of two states in $\mathcal{S}_a$ is accessible from the other state. Given this and considering the WA condition as in \cite[Definition~4.2.2]{bertsekas2007volume}, the modeled CMDP is weakly accessible.
\end{proof}
\vspace{-0.02cm}
As the WA condition holds for the modeled CMDP, it can be concluded that the optimal expected GoE remains the same for all initial states, i.e., $\operatorname{GoE}_n^{\rm (pull)}$, $\forall n$, is independent of $\operatorname{GoE}_0^{\rm (pull)} $\cite[Proposition~4.2.3]{bertsekas2007volume}. 
Moreover, Proposition~\ref{propose1} confirms the existence of an optimal stationary policy $\pi^*$ for $\mathcal{P}_1$. This optimal policy is \emph{unichain} \cite[Proposition~4.2.6]{bertsekas2007volume}. To solve the problem, we first relax its constrained form via defining a dual problem, and then we proceed with proposing an algorithm.

\subsection{Dual problem} We convert the constrained form of $\mathcal{P}_1$ to an unconstrained one via writing the Lagrange function $\mathcal{L}(\mu; \pi)$ as follows
\begin{align}\label{opt:opt2}
    \mathcal{L}(\mu; \pi) &=  \underset{N \rightarrow \infty}{\lim \sup} ~ \frac{1}{N}\, \mathbb{E}\bigg[\sum_{n = 1}^{N} \left(\operatorname{GoE}_n^{\rm (pull)} - \mu g_c(\alpha_n c_n)\right) \!\bigg]\! \nonumber \\
    & ~~~ + \mu C_{\rm max}
\end{align}
where $\mu \geq 0$ indicates the Lagrange multiplier. Then, we can summarize the Lagrange dual problem to be solved as
\begin{equation}\label{opt:opt3}
    \mathcal{P}_2 : ~ \underset{\mu \geq 0}{\inf}~ \underbrace{\underset{\pi}{\operatorname{max}}~\mathcal{L}(\mu; \pi)}_{\coloneqq\,   h(\mu)}
\end{equation}
in the form of an \emph{unconstrained} MDP, where $h(\mu)=\mathcal{L}(\mu; \pi_\mu^*)$ shows the Lagrange dual function. Here, $\pi_\mu^* \!: \mathcal{S} \rightarrow \mathcal{A}$ appears as the $\mu$-optimal policy and is derived from the dual problem for a given $\mu$, as $\pi_\mu^* =  \underset{\pi}{\operatorname{arg\,max}}~\mathcal{L}(\mu; \pi)$.

Since $\mathcal{S}$ has finite states, which ensures the growth condition as in \cite{altman1999constrained}, and $\operatorname{GoE}_n \geq 0$, $\forall n$, from \eqref{eq:eq1}, $\mathcal{P}_1$ and $\mathcal{P}_2$ converge to the same optimal values \cite[Corollary~12.2]{altman1999constrained}. Therefore, we can find the class of optimal policies $\pi^*$ after applying an \emph{iterative algorithm} approach\cite{hatami2022demand}. 

\subsection{Iterative Algorithm}
\SetKwFunction{FMain}{Utility} 
\SetKwProg{Fn}{Function}{:}{}
\begin{algorithm}[t!]
\DontPrintSemicolon
    \caption{Solution for deriving $\pi^*$ and $\mu^*$.} \label{Alg1}
    \KwInput{Known parameters $N \gg 1$, $c_n$, $\forall n$, $\varepsilon_\mu$, $C_{\rm max}$, $\eta$, states $\mathcal{S}$, and actions $\mathcal{A}$. Initial values $l \leftarrow 1$, $\mu^{(0)} \leftarrow 0$, $\mu^- \leftarrow 0$, $\mu^+ \gg 1$, $\pi_\mu^- \leftarrow 0$, and $\pi_\mu^+ \leftarrow 0$. The form of $g_c$.
    }
    Initialize $\pi_\mu^*(s)$, $\forall s \in \mathcal{S}$, via running \FMain{$\mu^{(0)}$}.\\
    \lIf{$\mathbb{E}\Big[ \sum_{n = 1}^{N} g_c(\alpha_n c_n)\Big] \leq NC_{\rm max}$}{\textbf{goto} {\scriptsize{\textbf{\ref{line:return1}}}}.}
    \While{$|\mu^+\!-\!\mu^-| \geq \varepsilon_\mu$}
    {
    \nonl\textit{Step} $l$: \Comment{Outer loop (Bisection search)}\\
    Update $\mu^{(l)} \leftarrow \frac{\mu^++\mu^-}{2}$ and $\pi_\mu^* \leftarrow$ \FMain{$\mu^{(l)}$}.\\
    \If{$\mathbb{E}\Big[ \sum_{n = 1}^{N} g_c(\alpha_n c_n)\Big] \geq NC_{\rm max}$}{$\mu^- \leftarrow \mu^{(l)}$, and $\pi_\mu^- \leftarrow$ \FMain{$\mu^-$}.}
    \lElse{$\mu^+ \leftarrow \mu^{(l)}$, and $\pi_\mu^+ \leftarrow$ \FMain{$\mu^+$}.}
    Reset $l \leftarrow l+1$.
    } 
    \lIf{$\mathbb{E}\Big[ \sum_{n = 1}^{N} g_c(\alpha_n c_n)\Big] < NC_{\rm max}$}{$\pi_\mu^*(s) \leftarrow \eta \pi_\mu^-(s) + (1\!-\!\eta)\pi_\mu^+(s)$, $\forall s \in \mathcal{S}$.}
    \KwRet $\mu^* = \mu^{(l)}$ and $\pi^*(s) = \pi^*_{\mu}(s)$, $\forall s\in\mathcal{S}$.\label{line:return1}\\
    \nonl
    \nonl\hrulefill\\
    \nonl
    \Fn{\FMain{$\mu$}}{
    \KwInput{Known parameters $N \gg 1$, $\varepsilon_v$, states $\mathcal{S}$, and actions $\mathcal{A}$. Initial values $t \leftarrow 1$, $\pi_\mu^*(s) \leftarrow 0$, and $V_{\pi_\mu^*}^{(0)}(s) \leftarrow 0$, $\forall s \in \mathcal{S}$.
    }
        \nonl\textit{Iteration} $t$: \Comment{Inner loop (Value iteration approach)}\\
        \For{state $s \in \mathcal{S}$\label{line:iter_t}}{compute $V_{\pi_\mu^*}^{(t)}(s)$ from \eqref{opt:opt5}\label{line:compute_v}, then update $\pi_\mu^*(s)$ as in \eqref{opt:opt5b} with the optimal action from $V_{\pi_\mu^*}^{(t)}(s)$.
        }
        \If{$\underset{s\in \mathcal{S}}{\operatorname{max}}  \big|V_{\pi_\mu}^{(t+1)}(s) \!-\!V_{\pi_\mu}^{(t)}(s)\big|\geq {\varepsilon_v}$}{step up $t \leftarrow t+1$, and \textbf{goto} {\scriptsize{\textbf{\ref{line:iter_t}}}}.}
        \KwRet $\pi_\mu^*$.
  }
\end{algorithm}

The iterative approach is illustrated in Algorithm~\ref{Alg1} with two \emph{inner} and \emph{outer} loops for deriving the $\mu$-optimal policy, i.e., $\pi_\mu^*$, and the optimal Lagrange multiplier, i.e., $\mu^*$, respectively.

\subsubsection{Computing $\pi_\mu^*$}
In the inner loop, with a given $\mu$ from the outer loop, the query policy is iteratively updated taking an optimal action which maximizes the \emph{expected utility} (value) $V_{\pi_\mu}^{(t)}(s)$ for the state $s \in \mathcal{S}$ at the $t$-th, $t\in\mathbb{N}$, iteration. Under the form of the value iteration approach for the unichain policy MDPs \cite{puterman2014markov}, the optimal value function is derived from Bellman's equation \cite{bellman1952theory}, as follows
\begin{equation}\label{opt:opt5}
    V_{\pi_\mu^*}^{(t)}(s) = \underset{\alpha \in \mathcal{A}}{\operatorname{max}} \sum_{s^\prime \in \mathcal{S}} P(s, \alpha, s^\prime) \!\left[R_\mu(s, \alpha, s^\prime) \!+\!  V_{\pi_\mu^*}^{(t-1)}(s^\prime)\right]\!.
\end{equation}
Consequently, the optimal policy for $s \in \mathcal{S}$ is updated by
\begin{equation}\label{opt:opt5b}
    \pi_\mu^*(s) \in \underset{\alpha \in \mathcal{A}}{\operatorname{arg\,max}} \sum_{s^\prime \in \mathcal{S}} P(s, \alpha, s^\prime) \!\left[R_\mu(s, \alpha, s^\prime) \!+\!  V_{\pi_\mu^*}^{(t-1)}(s^\prime)\right]\!
\end{equation}
where $R_\mu(s, \alpha, s^\prime) \coloneqq R(s, \alpha, s^\prime) - \mu g_c(\alpha c)$ is the \emph{net} reward.

The inner loop stops once the stopping/convergence criterion $\underset{s\in \mathcal{S}}{\operatorname{max}}\big|V_{\pi_\mu^*}^{(t+1)}(s) \!-\! V_{\pi_\mu^*}^{(t)}(s)\big| \leq \varepsilon_v$ is satisfied, where $\varepsilon_v$ indicates the convergence accuracy. As the query policies are unichain with aperiodic transition matrices, the above convergence criterion is reached for some finite iterations \cite[Theorem~8.5.4]{puterman2014markov}.

\subsubsection{Computing $\mu^*$}
In order to find the optimal Lagrange multiplier $\mu^{(l)}$ at the $l$-th, $l\in\mathbb{N}$, step of the outer loop, according to the updated $\pi_\mu^*$ from the inner loop, we apply the so-called bisection method with the stopping criterion $|\mu^+\!-\!\mu^-|\leq \varepsilon_\mu$ and the search accuracy $\varepsilon_\mu$, as depicted in Algorithm~\ref{Alg1}. From \eqref{opt:opt2} and \eqref{opt:opt3}, the increase of $\mu$ continuously increases the dual function $h(\mu)$ while decreasing the net reward $\operatorname{GoE}_n^{\rm (pull)} - \mu g_c(\alpha_n c_n)$ and the query arrival rate. Thus, we search for the smallest value of $\mu$ that satisfies the communication cost $C_{\rm max}$. Since $\operatorname{GoE}_n^{\rm (pull)}$ is independent of $\mu g_c(\alpha_n c_n)$, one can verify that $h(\mu)$ is a Lipschitz continuous function of $\mu$ with the  Lipschitz constant of 
\begin{equation*}
   \bigg|C_{\rm max} - \underset{N \rightarrow \infty}{\lim \sup} ~ \frac{1}{N}\, \mathbb{E}\bigg[\sum_{n = 1}^{N} g_c(\alpha_n c_n) \bigg]\bigg|.
\end{equation*}
Thus, the outer loop converges to the optimal multiplier after finite iterations \cite[pp.~294]{Wood2009}. The optimal value is attained based on a simple non-randomized stationary policy or a mix of two non-randomized policies with a mixing probability $\eta$, which can be obtained such that
$\mathbb{E}\big[ \sum_{n = 1}^{N} g_c(\alpha_n c_n)\big] = NC_{\rm max}$ \cite{beutler1985optimal}.

Algorithm~\ref{Alg1} has at most $\mathcal{O}(2LT\Delta^2_{\rm max}|\mathcal{V}|^2)$ arithmetic operations, where $L$ and $T$ indicate the step sizes of the outer and inner loops, respectively. The increase of the states and the iteration size of either loop increase the algorithm’s complexity. 
Nevertheless, this complexity is manageable in real-world scenarios owing to the following items:
\begin{itemize}
\renewcommand{\labelitemi}{\tiny$\blacksquare$}
    \item $\Delta_{\rm max}$ is set to a small value since the effectiveness of an update at the endpoint saturates to a close to zero value past a level of staleness, making that update not useful.
    \item As the usefulness of the updates can be normalized,i.e., $\nu_i \in [0, 1]$, $\forall i \in \mathcal{V}$, a large outcome space is not necessarily needed for $\mathcal{V}$.
    \item The communication cost is usually fixed, thus reducing the bisection search interval, hence $L$, after some trials.
\end{itemize}

\subsection{Threshold-Based Query Control Model}\label{Sec:Sec4d}
Finding the optimal policy $\pi^*$ from Algorithm~\ref{Alg1} gives us a threshold criterion $\Omega_{\rm th}$ for every CMDP's state, following which maximizes the GoE of the system subject to the communication cost. With $\Omega_{\rm th}$ in hand, the query controller timely decides to pull an update or not based on the current AoI, i.e., $\Delta_n$, and the meta-value of the latest correctly received update at the query time, i.e., $v_m$ for $m = \underset{i:\,\alpha_i(1-\epsilon_i)=1,\, i \leq n}{\rm {max}~} i$. In this sense, the optimal action $\alpha_n^*$ for the $n$-th time slot can be derived according to two alternative options, as given below.
\begin{enumerate}[leftmargin=5mm,align=left]
    \item[\textbf{{Option~I:}}] The value of $\Delta_{\rm th}$ is a function of $v_m$, where
    \begin{equation}\label{opt:opt7}
    \Omega_{\rm th}  = \Delta_{\rm th}^{(v_m)} \text{,~hence~} \alpha_n^* = \mathbbm{1}(\Delta_n \geq \Delta_{\rm th}^{(v_m)} \,|\, v_m).
    \end{equation}
    \item[\textbf{{Option~II:}}] The level of $v_{\rm th}$ depends on $\Delta_n$. Thus, we have
    \begin{equation}\label{opt:opt8}
     \Omega_{\rm th}  = v_{\rm th}^{(\Delta_n)} \text{,~hence~} \alpha_n^* = \mathbbm{1}(v_m \leq v_{\rm th}^{(\Delta_n)} \,|\, \Delta_n).
    \end{equation}
\end{enumerate}

\section{Simulation Results}\label{Sec:Sec5}
We assess the performance of the proposed effect-aware query control policy within $N= 1000$ time slots and compare it with three existing effect-agnostic query arrival models, namely (i) \emph{periodic} model, and stochastic models following (ii) \emph{binomial} and (iii) \emph{Markovian} process. For the latter model, we have a Markov chain with two states of ``pulling an update" and ``keeping silent", in which the self-transition probability of the latter state is $0.95$, while the one for the primary state relies on the query rate. The default query rate for the effect-agnostic policies is $0.8$ unless otherwise stated. 
Furthermore, we equally divide the interval $[0, 1]$ into $|\mathcal{V}| = 10$ levels and initialize the importance set $\mathcal{V} = \{0, 0.11, 0.22, 0.33, 0.44, 0.56, 0.67, 0.78, 0.89, 1
\}$, each indicating a normalized rank of importance with $p_{ii}=p_{ij}=\frac{1}{|\mathcal{V}|}$, $\forall \nu_i,\nu_j\in \mathcal{V}$. The maximum acceptable AoI is $\Delta_{\rm max}=10$, and the probability of the update discrepancy is $p_\epsilon^{(n)}=0.2$, $\forall n$. Also, for Algorithm~\ref{Alg1}, we set $\varepsilon_v = \varepsilon_\mu = 10^{-3}$ and $\eta=0.5$.

For performance evaluation, we define a \emph{net} GoE (NGoE) metric, which incorporates the GoE and the cost, as follows
\begin{align}\label{eq:eqmetric}
    \operatorname{NGoE}_n^{\rm (pull)} &= \operatorname{exp}\!\big(\!\operatorname{GoE}_n^{\rm (pull)} \!- g_c(\alpha_n c_n)\big) \nonumber \\
    & = \operatorname{exp}\!\big(\!-\!v_n\Delta_n \!-\! c_0\alpha_n \big)
\end{align}
for the $n$-th slot, $\forall n\in \mathcal{N}$, where the exponential form is arbitrarily used to guarantee positiveness. In \eqref{eq:eqmetric}, without loss of generality, we employ a \emph{linear} form for $g_\Delta$, $g_v$, and $g_c$, and a \emph{multiplication} form for $f$ according to \eqref{eq:eq1}. Besides, we consider \emph{uniform} cost $c_n = c_0$, $\forall n$, and unless otherwise specified, we initialize $c_0=0.5$ and $C_{\rm max} = 0.4$.

\begin{figure}[t!]
    \centering
    \pstool[scale=0.43]{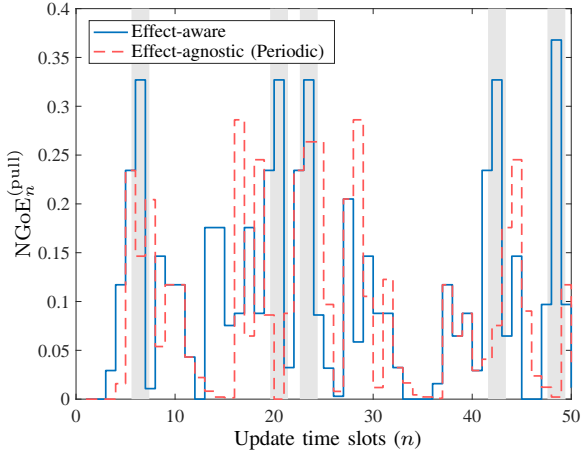}{
    \psfrag{Value-aware}{\hspace{0.00cm}\scriptsize Effect-aware}
    \psfrag{Value-agnostic(Periodic)------}{\hspace{0.00cm}\scriptsize Effect-agnostic (Periodic)}
    \psfrag{Effectivenessdegree}{\hspace{0.35cm}\footnotesize $\operatorname{NGoE}_n^{\rm (pull)}$}
    \psfrag{Updatetimeslots}{\hspace{-0.32cm}\footnotesize Update time slots ($n$)}
    }
    \vspace{-0.2cm}
    \caption{The normalized effectiveness performance of different query control policies within $N=50$ slots.} 
    \label{fig:fig2}
\end{figure}
Fig.~\ref{fig:fig2} depicts a $50$-slot snapshot of the status update system according to the arrived queries and the updates' NGoE, comparing the normalized GoE of the proposed effect-aware with the periodic effect-agnostic query control policy. The latter policy is assumed to have $7$-slot period intervals. We observe that our effect-aware policy enables capturing updates with the highest NGoE, of which the horizons are highlighted. However, in some slots, the effect-aware policy does not pull usefulness updates due to the probabilistic uncertainties of the CMDP problem, while the periodic one can catch them by chance. 

\begin{figure}[t!]
    \centering
    \pstool[scale=0.43]{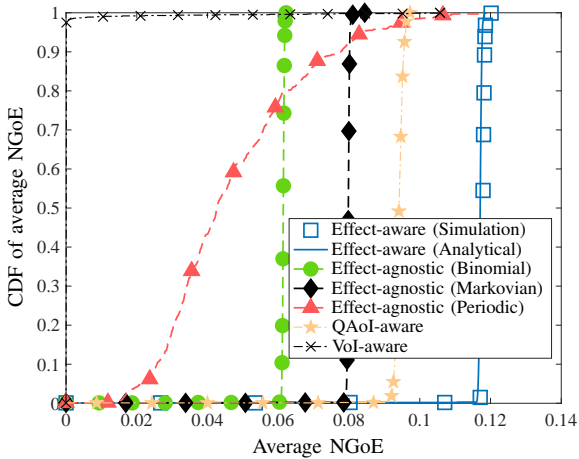}{
    \psfrag{Value-aware(Simulation)}{\hspace{0.00cm}\scriptsize Effect-aware (Simulation)}
    \psfrag{Value-aware(Analytical)}{\hspace{0.00cm}\scriptsize Effect-aware (Analytical)}
    \psfrag{Value-agnostic(Periodic)}{\hspace{0.00cm}\scriptsize Effect-agnostic (Periodic)}
    \psfrag{Value-agnostic(Binomial)}{\hspace{0.00cm}\scriptsize Effect-agnostic (Binomial)}
    \psfrag{Value-agnostic(Markovian)------}{\hspace{0.00cm}\scriptsize Effect-agnostic (Markovian)}
    \psfrag{Specialcase(QAoI)}{\hspace{0.00cm}\scriptsize QAoI-aware}
    \psfrag{Specialcase(VoI)}{\hspace{0.00cm}\scriptsize VoI-aware}
    \psfrag{Effectivenessdegree}{\hspace{0.26cm}\footnotesize Average NGoE}
    \psfrag{CDFofeffectivenessdegree}{\hspace{0.1cm}\footnotesize CDF of average NGoE}
    }
    \vspace{-0.1cm}
    \caption{The CDF of the average NGoE provided over $N=1000$ time slots.}
    \label{fig:fig3}
\end{figure}
Fig.~\ref{fig:fig3} presents the cumulative distribution function (CDF) of the average NGoE provided by applying the effect-aware and effect-agnostic policies over $N=1000$ time slots. It is shown that the effect-aware policy highly boosts the effectiveness of the system, thanks to its prediction of the updates’ usefulness and through pulling significant updates, compared to the effect-agnostic policies. Specifically, the effect-aware policy increases the effectiveness by $91\%$, $47\%$, and $149\%$ on average, as compared to the binomial, Markovian, and periodic models, respectively. This comes at the cost of $16\%$ on average higher transmission rate for the effect-aware policy than the others. 

\begin{figure}[t!]
    \centering
    \subfloat[]{
    \pstool[scale=0.43]{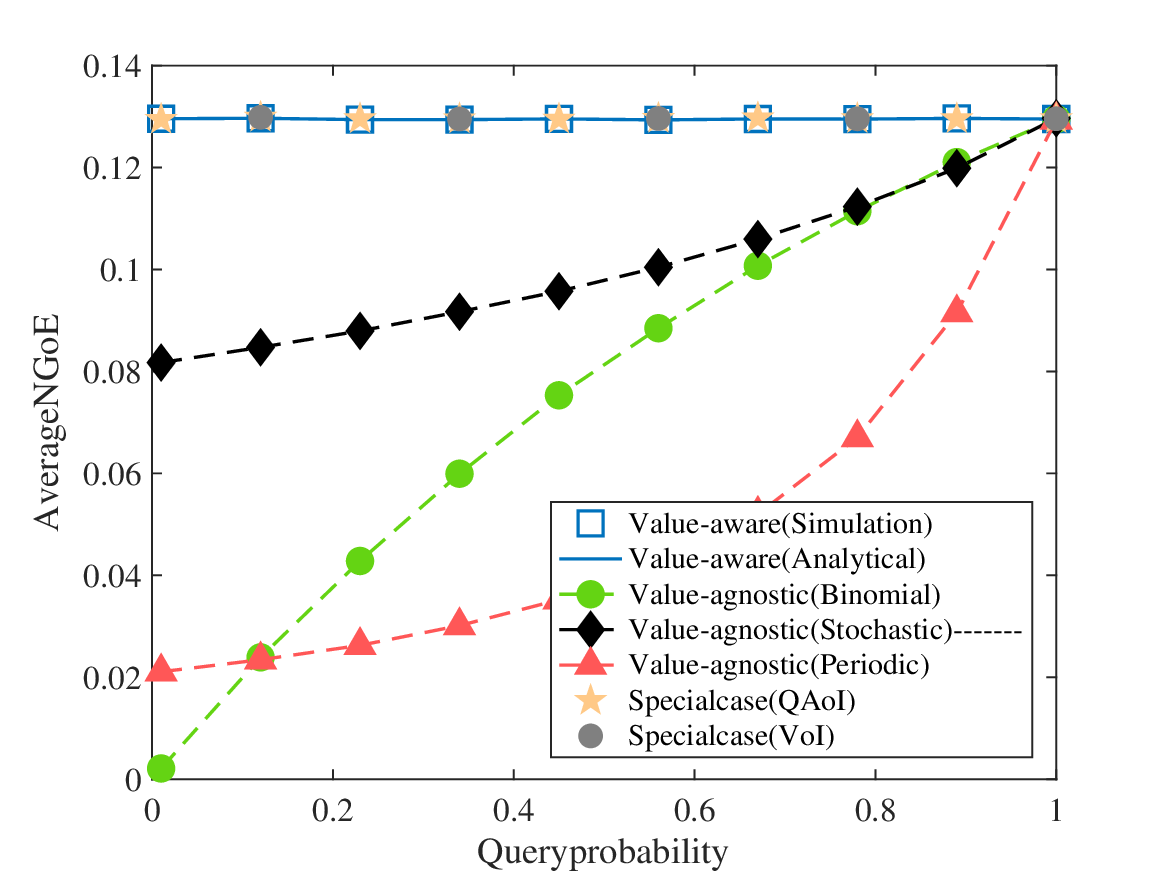}{
    \psfrag{Value-aware(Simulatios)-----}{\hspace{0.00cm}\scriptsize Effect-aware (Simulation) $c_0=0.1$}
    \psfrag{Value-aware(Simulation)}{\hspace{0.00cm}\scriptsize Effect-aware (Simulation)}
    \psfrag{Value-aware(Analytical)}{\hspace{0.00cm}\scriptsize Effect-aware (Analytical)}
    \psfrag{Value-agnostic(Periodic)}{\hspace{0.00cm}\scriptsize Effect-agnostic (Periodic)}
    \psfrag{Value-agnostic(Binomial)}{\hspace{0.00cm}\scriptsize Effect-agnostic (Binomial)}
    \psfrag{Value-agnostic(Stochastic)-------}{\hspace{0.00cm}\scriptsize Effect-agnostic (Markovian)}
    \psfrag{Specialcase(QAoI)}{\hspace{0.00cm}\scriptsize QAoI-aware}
    \psfrag{Specialcase(VoI)}{\hspace{0.00cm}\scriptsize VoI-aware}
    \psfrag{AverageNGoE}{\hspace{-0.06cm}\footnotesize Average NGoE}
    \psfrag{Queryprobability}{\hspace{-0.29cm}\footnotesize Controlled query rate}
    }}
    \hfill
    \vspace{-0.005cm}
    \centering
    \subfloat[]{
    \pstool[scale=0.43]{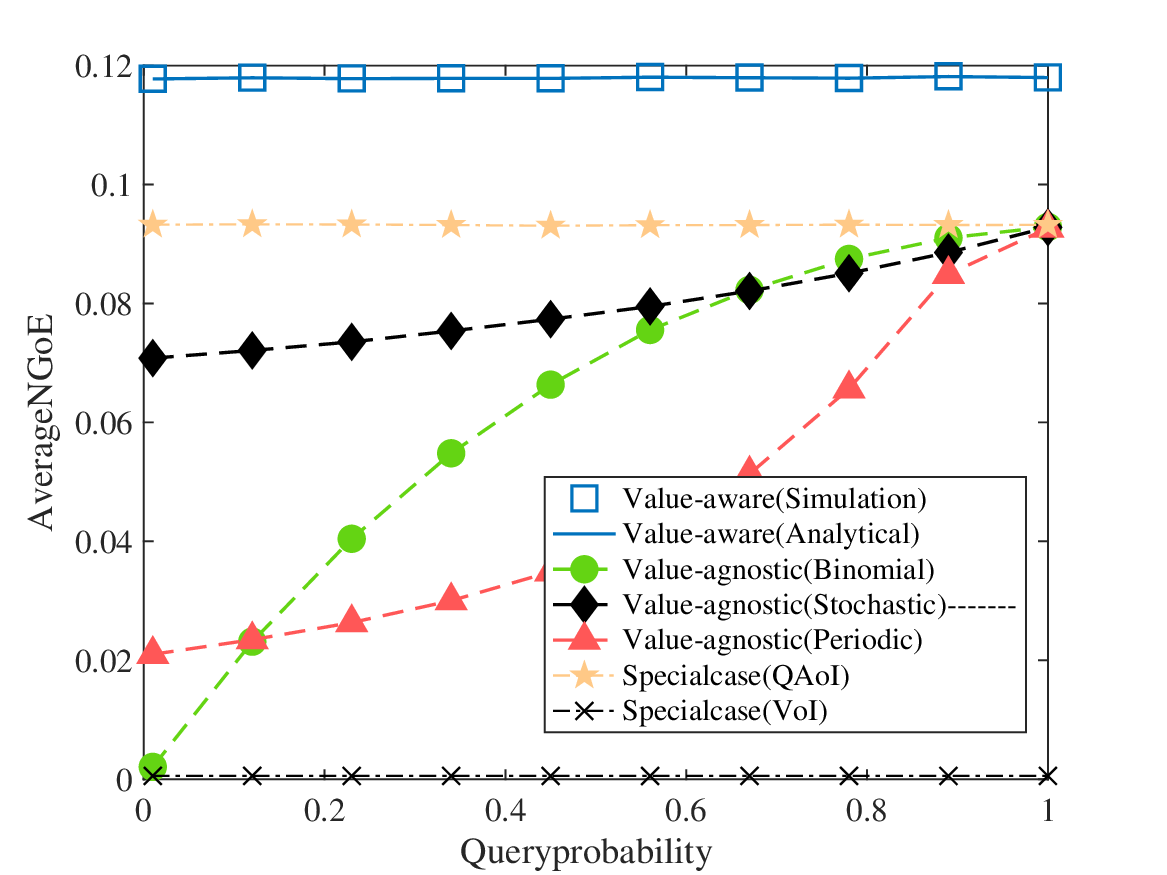}{
    \psfrag{Value-aware(Simulation)}{\hspace{0.00cm}\scriptsize Effect-aware (Simulation)}
    \psfrag{Value-aware(Analytical)}{\hspace{0.00cm}\scriptsize Effect-aware (Analytical)}
    \psfrag{Value-agnostic(Periodic)}{\hspace{0.00cm}\scriptsize Effect-agnostic (Periodic)}
    \psfrag{Value-agnostic(Binomial)}{\hspace{0.00cm}\scriptsize Effect-agnostic (Binomial)}
    \psfrag{Value-agnostic(Stochastic)-------}{\hspace{0.00cm}\scriptsize Effect-agnostic (Markovian)}
    \psfrag{Specialcase(QAoI)}{\hspace{0.00cm}\scriptsize QAoI-aware}
    \psfrag{Specialcase(VoI)}{\hspace{0.00cm}\scriptsize VoI-aware}
    \psfrag{AverageNGoE}{\hspace{-0.06cm}\footnotesize Average NGoE}
    \psfrag{Queryprobability}{\hspace{-0.29cm}\footnotesize Controlled query rate}
    }
    }
    \caption{The interplay between average NGoE and query rate for different query control policies and cost coefficients (a) $c_0=0.1$ and (b) $c_0=0.5$.}
    \label{fig:fig4}
\end{figure}
Figs.~\ref{fig:fig4}\,(a) and \ref{fig:fig4}\,(b) show the interplay between average NGoE and controlled query rate for $c_0=0.1$ and $c_0=0.5$, respectively, for effect-aware and effect-agnostic policies, as well as for two special cases relevant to QAoI and VoI within $N=1000$ time slots. In the QAoI-aware case, decision policies are obtained to maximize the long-term expected QAoI regardless of the usefulness or importance of the updates, whereas in the VoI-aware case, decision policies depend on the expected VoI, without taking freshness into account. We can see that the effect-aware policy outperforms the other effect-agnostic policies for all query rates and under both cost values.
Comparing Fig.~\ref{fig:fig4}\,(a) with Fig.~\ref{fig:fig4}\,(b), we infer that, for $c_0=0.1$, the special cases offer almost the same performance as the effect-aware policy considering both freshness and usefulness attributes. This is because all policies can pull the same updates, thus having identical performance. However, for $c_0=0.5$, the effect-aware policy outperforms both special cases for all query rates, even using fewer resources, owing to its effect-aware query control. Also, since $\Delta_n\geq1$, and $0\leq v_n\leq1$, $\forall n$, bypassing the freshness attribute in the GoE metric leads to a significant performance drop. Besides, the increase of the cost decreases the offered average NGoE, and that at a faster speed for the effect-agnostic policies. Indeed, not useful or irrelevant transmissions under the effect-agnostic policies result in a higher waste of resources and a larger effectiveness gap.

To study the convergence of the iterative approach for solving the CMDP problem regarding to Algorithm~\ref{Alg1}, we plot Fig.~\ref{fig:fig5}, which illustrates the value (expected utility) obtained in each iteration for different cost coefficients. It can be seen that the policy convergences to its final value after $130$, $133$, and $137$ iterations, sequentially, for the cost coefficients $c_0=0.1$, $c_0=0.5$, and $c_0=1$. We can also observe that the value provided in each iteration decreases by increasing the cost, in line with the formulation of the CMDP's rewards in Section~\ref{Sec:Sec4a}.
\begin{figure}[t!]
    \centering
    \pstool[scale=0.43]{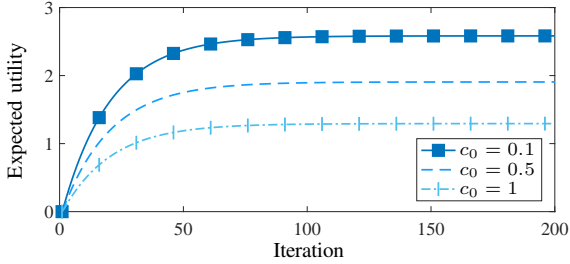}{
    \psfrag{cost=0.1}{\hspace{0.00cm}\scriptsize $c_0 = 0.1$}
    \psfrag{cost=0.5----}{\hspace{0.00cm}\scriptsize $c_0 = 0.5$}
    \psfrag{cost=1}{\hspace{0.00cm}\scriptsize $c_0 = 1$}
    \psfrag{IterationValue}{\hspace{-0.15cm}\footnotesize Expected utility}
    \psfrag{Iteration}{\hspace{-0.0cm}\footnotesize  Iteration}
    }
    \vspace{-0.6cm}
    \caption{The expected utility obtained in each iteration of Algorithm~\ref{Alg1}.}
    \label{fig:fig5}
\end{figure}

\noindent
\renewcommand{\arraystretch}{1}
\begin{table}[t!]
\begin{center}
\caption{Illustration of the threshold boundaries for the threshold-based query control model}\label{tab:tab1}
\vspace{-0.11cm}
\begin{tabular}{| l | l |  c | c | c | c | c | c | c |}
\hline
\multicolumn{3}{|c|}{\!\!\small $v_m$\!\!}&\multicolumn{1}{c|}{\small \!$0 \rightarrow 0.44$\!}&\multicolumn{1}{c|}{\small \!$0.56$\!}&\multicolumn{1}{c|}{\small \!$0.67$\!}&\multicolumn{1}{c|}{\small \!$0.78$\!}&\multicolumn{1}{c|}{\small \!$0.89$\!}&\multicolumn{1}{c|}{\small \,~$1$~\, }\\
\hline
\hline
\!\!\parbox[t]{3.0mm}{\multirow{3}{*}{\rotatebox[origin=c]{90}{\small $c_0\!=\!0.5$}}} \!\!\!\!&\!\!\parbox[t]{3.0mm}{\multirow{3}{*}{\rotatebox[origin=c]{90}{\small $\Delta_n$}}} \!\!\!\!&\multicolumn{1}{c|}{\small $1 \rightarrow 2$}&\cellcolor{pull}&\multicolumn{5}{c|}{\cellcolor{silent}{\footnotesize Silent}}  \\
\cline{3-3}
\hhline{*{3}{|~}>{\arrayrulecolor{pull}}->{\arrayrulecolor{black}}*{5}{|-}}
&&\multicolumn{1}{c|}{\small $3$}&\cellcolor{pull}&\multicolumn{2}{c|}{\cellcolor{pull}}&\multicolumn{3}{c|}{\cellcolor{silent}\footnotesize Silent} \\
\cline{3-3}
\hhline{*{3}{|~}>{\arrayrulecolor{black}}|>{\arrayrulecolor{pull}}->{\arrayrulecolor{black}}|>{\arrayrulecolor{pull}}-->{\arrayrulecolor{black}}*{3}{|-}}
&&\multicolumn{1}{c|}{\small \!$4 \rightarrow 10$\!}&\cellcolor{pull}\multirow{-3}{*}{\footnotesize Pull}&\multicolumn{2}{c|}{\cellcolor{pull}\multirow{-2}{*}{\footnotesize Pull}}&\multicolumn{3}{c|}{\cellcolor{pull}\footnotesize Pull}\\
\hline
\hline
\!\!\parbox[t]{3.0mm}{\multirow{2}{*}{\rotatebox[origin=c]{90}{\small $c_0\!=\!1$}}} \!\!\!\!&\!\!\parbox[t]{3.0mm}{\multirow{2}{*}{\rotatebox[origin=c]{90}{\small $\Delta_n$}}} \!\!\!\!&\multicolumn{1}{c|}{\small $1 \rightarrow 3$}&\cellcolor{pull}&\multicolumn{5}{c|}{\cellcolor{silent}{\footnotesize Silent}} \\
\cline{3-3}
\hhline{*{3}{|~}>{\arrayrulecolor{black}}|>{\arrayrulecolor{pull}}->{\arrayrulecolor{black}}*{5}{|-}}
&&\multicolumn{1}{c|}{\small \!$4 \rightarrow 10$\!}&\cellcolor{pull}\multirow{-2}{*}{\footnotesize Pull}&\multicolumn{5}{c|}{\cellcolor{pull}\footnotesize Pull}\\
\hline
\end{tabular}
\medskip
\end{center}
\vspace{-0.3cm}
\end{table}
Finally, Table~\ref{tab:tab1} demonstrates the threshold boundaries for the threshold-based query control model discussed in Section~\ref{Sec:Sec4d} for cost coefficients $c_0=0.5$ and $c_0=1$. The term ``Silent" corresponds to $\alpha_n=0$, and ``Pull" indicates $\alpha_n=1$. Either Option~I or Option~II could be applied to find $\Omega_{\rm th}$ based on $\Delta_{\rm th}^{(v_m)}$ or $v_{\rm th}^{(\Delta_n)}$, respectively. As in Table~\ref{tab:tab1}, the query controller pulls updates in case $v_m \leq 0.44$, $\forall m$, regardless of the current AoI. Also, updates are always pulled if $\Delta_n \geq 4$, $\forall n$, independent of the degree of usefulness. For the above conditions, we thus have fixed $\Delta_{\rm th} = 4$ and $v_{\rm th} = 0.44$ under any cost coefficient for the primary and latter options, respectively. In order to derive the optimal action for the other conditions (CMDP’s states), the communication cost plays a key role, hence resulting in variable threshold metrics. As the cost increases, the controller should pull updates merely under more \emph{critical} conditions to increase the system's NGoE. 

\vspace{-0.05cm}
\section{Conclusion}\label{Sec:Sec6}
We have proposed an effect-aware query control policy for pull-based communication systems, in which the query controller timely decides whether to pull an update, depending on the source's evolution and the updates' effectiveness at the endpoint. We have considered the problem of GoE maximization based on a CMDP with finite state spaces for the AoI and the usefulness rank and have provided an algorithm to find the class of optimal policies. Our results have shown that the effect-aware query policy could provide significant gains in terms of normalized GoE compared to effect-agnostic policies for different query rates and communication costs.

\vspace{-0.05cm}
\bibliographystyle{IEEEtran}
\bibliography{References.bib}
\balance

\end{document}